\documentclass{article}

\usepackage{fancyhdr}
\date{}  

\usepackage{url}
\usepackage{cite}
\usepackage{plain}

\usepackage{wrapfig}
\usepackage{graphics}
\usepackage{picins,graphicx}
\usepackage[english]{babel}
\usepackage[leqno]{amsmath}
\usepackage{amssymb}
\usepackage{verbatim}
\usepackage[mathscr]{eucal}

\usepackage{amstext}
\usepackage{makeidx}
\usepackage{amsthm}

\usepackage[bookmarks,colorlinks]{hyperref}
\hypersetup{colorlinks=false}

\usepackage[bookmarks,colorlinks]{hyperref}
\hypersetup{colorlinks=false}

\usepackage{hyperref}
\hypersetup{colorlinks,%
citecolor=red,%
filecolor=black,%
linkcolor=blue,%
urlcolor=blue}

\makeindex

\newtheorem{theorem}{Teorema} 

\newtheorem{lemma}{Lemma}
\newtheorem{proposition}{Proprieta'}
\newtheorem{definition}{Definizione}
\newtheorem{notation}{Nota}
\newtheorem{ex}{Esercizio} 
\newtheorem{esempio}{Esempio}

\newcommand{\beq}{\begin{equation}} 
\newcommand{\eeq}{\end{equation}}

\newcommand{\bex}{\begin{ex}} 
\newcommand{\eex}{\end{ex}} 

\newcommand{\bese}{\begin{esempio}} 
\newcommand{\eese}{\end{esempio}} 

\newcommand{\bpro}{\begin{proposition}} 
\newcommand{\epro}{\end{proposition}}

\newcommand{\bthe}{\begin{theorem}} 
\newcommand{\ethe}{\end{theorem}}

\newcommand{\bnote}{\begin{notation}} 
\newcommand{\enote}{\end{notation}}

\newcommand{\bdefi}{\begin{definition}} 
\newcommand{\edefi}{\end{definition}} 

\newcommand{\bc}{\begin{center}} 
\newcommand{\ec}{\end{center}}

\newcommand{\mail}[1]{\href{unina:#1}{\texttt{#1}}}

\usepackage{pdfsync}

\author{Monica De Angelis\thanks{Univ. di Napoli  "Federico II", Scuola Politecnica e delle Scienze di Base, Dip. Mat. Appl. "R.Caccioppoli", \newline
 Via Cinthia 26, 80126, Napoli , Italia.
\newline\mail{modeange@unina.it}}}

\title{On the transition  from parabolicity to hyperbolicity for a nonlinear equation}
\begin{document}

\title{On the transition  from parabolicity to hyperbolicity for a nonlinear equation  under Neumann boundary conditions
}

\author{Monica De Angelis\thanks{Univ. di Napoli  "Federico II", Scuola Politecnica e delle Scienze di Base, Dip. Mat. Appl. "R.Caccioppoli", \newline
 Via Cinthia 26, 80126, Napoli , Italia.
\newline\mail{modeange@unina.it}}}

\author{Monica De Angelis \\Univ. di Napoli  "Federico II",\\ Scuola Politecnica e delle Scienze di Base,\\ Dip. Mat. Appl. "R.Caccioppoli", \\
 Via Cinthia 26, 80126, Napoli , Italia.
\\ \mail{modeange@unina.it}
}


\maketitle

An integro differential equation  which is able  to describe the evolution of a large class of dissipative models, is considered.   By means of an equivalence, the focus  shifts to  the perturbed sine-Gordon equation  that in  superconductivity finds interesting applications in multiple engineering  areas. The Neumann boundary problem is considered, and  the behaviour  of a viscous
term, defined by a higher-order derivative with small diffusion coefficient $ \varepsilon, $ is  investigated.  The Green function, expressed by means of  Fourier series, is considered,  and an estimate is  achieved. Furthermore, some classes of solutions of the hyperbolic equation are determined, proving that there exists at least one solution with bounded derivatives.
Results obtained prove that  diffusion effects 
 are bounded and tend to zero when $ \varepsilon $  tends to zero.



\section{Introduction}
\label{intro}


Evolution problems of several physical models such  as motions of viscoelastic bodies or fluids,\cite{o,re},  heat conduction, \cite{ru,bw}, sound propagation \cite{ss}, and  biological systems such  as neural communications \cite{drw,dew13,drultimo,m2},  can be modelled by means of the following integro differential equation: 
\begin{equation}   \label{11}
  \mathcal L u \equiv \,\, u_t -  \varepsilon  u_{xx} + au +\delta \int^t_0  e^{- \beta (t-\tau)}\, u(x,\tau) \, d\tau \,=\, F(x,t,u) \, 
  \end{equation} 

\noindent  which   is part of a more general class of mathematical physical models typical of materials with memory
 \cite{sc,es,as,ccvc}(and references therein). 
 
 Furthermore,  equation (\ref{11}) can also  describe specific  phenomena  in  superconductivity (see, f.i.\cite{mda2010,acscott02,uffa,32}). Indeed,  when in  (\ref{11}) one assumes

\begin{equation}   
 a \,=\,  \alpha \, - \frac{1}{\varepsilon} \, \,\quad\,\, \delta = \,  - \, \frac{a}{\varepsilon}  \,\,\quad \displaystyle \, \beta \,= \frac{1}{\varepsilon}\,\,
  \end{equation}
  
  \noindent  and $ F  $ is such that

\begin{equation}   \label{16}
F(x,t,u)\,=\, -\, \int _0^t \, e^{\,-\,\frac{1}{\varepsilon}\,(t-\tau\,)}\,\,[\, \, sen \, u (x, \tau)\,- \gamma\, \,]\, \, d\tau, 
\end{equation} 
  integral equation (\ref{11}) turns into the  so called perturbed sine-Gordon  equation

  \begin{equation}  \label{17}
u_{xx}- u_{tt} =  \sin u + \gamma  + \alpha u_t - \varepsilon \,u_{xxt}.
\end{equation}

 Equation (\ref{17})  models the flux dynamics in the Josephson junction where two  superconductors are separated by a thin insulating layer. Term  $\alpha  u_t$   denotes the 
 dissipative normal electron current flow across the junction, while the third term  such as  $\varepsilon \,u_{xxt}$ represents the flowing  of quasi-particles  parallel to  the junction 
 \cite{sco}. The value range for $\alpha $ and
 $\varepsilon$ depends  on the material of the  real junction and in the vast majority of cases it results  $0<\alpha ,\varepsilon  < 1$   \cite{bp,pfssu}. 
 If, however, the resistance of the junction is so low  as to completely shorten the  capacitance, the case   $ \alpha $ large with respect to 1 arises \cite{ca}.

Josephson junctions find application in many fields of practical interest. For example, it is possible to diagnose heart and/or blood circulation problems using magnetocardiograms employing superconducting quantum interference devices (SQUIDs). Another area where SQUIDs are used is magnetoencephalography -MEG- where they allow to make inferences and analyze magnetic fields generated by electric currents inside brains.
  In geophysics, they are used as gradiometers  or as gravitational wave detectors (\cite{ca} and reference therein) and  play an important role in the study of the  potential virtues of  superconducting digital electronics \cite{f}. SQUIDs are  also used in nondestructive testing as a convenient alternative to ultra sound or x-ray methods (\cite{bp,bpagano,cla,cla2} and reference therein) and  can be
used as fast, switchable meta-atoms \cite{jbm}, too.
Besides,  superconductor technology and its applications are in continuous development. As an example, superconductivity powers the Tevatron collider at Fermilab to produce the top quark, and CERN’s Large Hadron Collider (LHC) to discover the Higgs boson and  advanced superconducting magnets are already being developed for future collider projects \cite{cern}.

 Of course, previously mentioned  equations are to be complemented by initial-boundary problems (many examples can be found, f.i. in \cite{dew13,m13,df13}) and a  particular attention will be given to the Neumann problem since it is meaningful in several scientific fields. In mathematical biology,  for instance, those boundary  conditions occur when  a two-species reaction diffusion system is subjected to flux boundary conditions \cite{m2}. The same conditions  are present in case of pacemakers \cite{ks} and they are applied  to study distributed FitzHugh-Nagumo systems as well \cite{ns}.  Moreover, in superconductivity, in the case of a Josephson junction, they  can refer to  the phase gradient value and they are  proportional to the magnetic field    \cite{for,j05}.

\noindent   

\subsection{Mathematical considerations, state of  the art and aim of the  paper}    

The  connection between  integro differential equation  (\ref{11}) and equation (\ref{17})   allows us  to extend the analysis performed on one model to the other. In particular, our attention will be devoted to the  parabolic third order linear   operator

\begin{equation}   \label{18}
{\cal L} u \equiv \,\,\partial_{xx} \,(\varepsilon
u  _{t}+u) - \partial_t( u_{t}+\alpha\,u)\,
\end{equation}

\noindent which characterizes equation (\ref{17}). Since in many cases the  evolution is  described by deep interactions between wave propagation and diffusion, operator      
        ${\cal L} $  can  also be considered  as a linear hyperbolic operator perturbed by viscous terms that are described by higher-order derivatives with small diffusion coefficients $ \varepsilon $.  Indeed, when  $ \varepsilon \equiv 0,$  the operator (\ref{18}) turns into a hyperbolic one: 

\begin{equation} \label{n19}
{\cal L}_0 U \equiv \,\,
U_{xx}   - \partial_t( U_{t}+\alpha\,U)\,
\end{equation}
and the influence of the dissipative terms, represented by $ \varepsilon \partial_{xxt} $ on the wave behavior has to   be estimated.

From this small  parameter depends   the transmission times of the signals as well, at least when the main process is the diffusion one. In order to  obtain some knowledge about the dynamics, the order of diffusion effects and the time-intervals must be estimated. This analysis leads to an evaluation of   the non linear parabolic - hyperbolic boundary layers. 

Similar problems, for Dirichlet conditions, have already been studied.  In particular, in  \cite{dmr},  for  $ \alpha=0, $ an asymptotic approximation was established by means of the  two characteristic times: slow time $ \tau = \varepsilon \,t $  and fast time $ \theta =\, t/\varepsilon $. Moreover, for  equations of type  (\ref{17}), in \cite{df213}  an analytical  analysis   proved      that the  surface damping has little influence
on the behaviour of the oscillator,  confirming numerical results showed  in  \cite{bcs00}. Other numerical investigations on  the influence of surface losses can be found in \cite{pskm}, too.  

As for the  Neumann problem, fast and slow diffusion effects have been investigated in \cite{mda18} and furthermore, an estimate allows us to  evaluate the behaviour when time t tends to infinity.

Moreover,     in \cite{df13,mda12,dri,ddf} a class of equations characterized by an operator like (\ref{18}) is considered and theorems of existence and uniqueness of solutions related to Dirichlet, Neumann and pseudoperiodic boundary conditions are proved.

Finally, we would like to point out that  the  linear operator $ \cal{L} $ has already been 	examined  by means of  a convolutions of Bessel
functions and  Fourier series,   and in this case,  the related Green Function $ G=\,\sum_{n=1}^{\infty}\, G_n $ has been estimated by means of constants \emph{  also depending  on parameter $ \varepsilon $ } \cite{mda12,mda}\emph{ or   considering} $G=\,\sum_{n=1}^{\infty}\, G_n/n^2$ \cite{mda18}.


Our  purpose is to   analyse  the \emph{error} $ d(x,t,\varepsilon) $  represented by the difference between  the solution $ u  $ related to the Neumann boundary problem $ \cal{P_{\varepsilon} }$  for  equation (\ref{17}) and  the wave solution $ U $ deduced from $ \cal{P_{\varepsilon}} $ when $ \varepsilon \rightarrow 0:  $ 

\begin{equation} \label{n10}
d(x,t,\varepsilon) = u(x,t,\varepsilon)-U(x,t).
\end{equation}   
 
For this purpose,  it appears necessary to determine  an estimate for  the Green function  proving that it is  bounded  by terms {\em independent from $ \varepsilon $} and vanishing when $ t \rightarrow \infty. $ 
 
 Moreover, some classes of solutions of the hyperbolic equation 
 \begin{equation}
 U_{xx} -
 \partial_t(U_{t}+\alpha\,U)\,=\sin U  +\gamma \label{19}
\end{equation}
 
\noindent will be determined showing  that there exists at least one solution with bounded  derivatives.

Finally, by applying a Gronwall inequality, it will be possible to highlight the infinite time-interval where the diffusion effects vanish for  $ \varepsilon \rightarrow. \,0 $

The paper is organized as follows:

In section 2,  the mathematical problem is stated  and  thanks to the Green Function, determined through   Fourier series,  the solution  related to the remainder term  $ d $  is shown. 

 In section 3, by means of  lemmas     on  hyperbolic and circular terms, in theorem \ref{th3} the following estimate  $H=\,\sum_{n=1}^{\infty}\, e^{-h_nt}\,\,  \frac{\sinh\,(\omega_nt)}{\omega_n}\,\leq\,B \, e ^{-{m\,\, t }}  $ is proved where $ B$ and $m $ are two  positive constants independent from $ \varepsilon  .$

In section 4,  many classes of  solutions of   (\ref{19}) are determined, and an example of solution  
 with bounded derivatives is shown.

 Finally, in section 5,  an
estimate for the remainder term is achieved by proving that the diffusion effects are of the order of $\,\displaystyle \varepsilon ^{1-k/\alpha} \log \varepsilon ^{-k} \,$   with $\, 0<\,k\,<\alpha. $

\section{Statement of the problem and solution of the nonlinear problem}       

 \label{2}

Letting $T$ be an arbitrary positive     constant and letting 
\begin{center} $  \Omega =\{(x,t) :  0\leq x\leq \ell 
, \  \ 0 \leq t \leq T \}, $\end{center}

\noindent we point our attention to the following Neumann problem:

  \begin{equation} \label{21}
  \left \{
  \displaystyle
    \begin{array}{ll}
      \partial_{xx}(\varepsilon
u _{t}+ u) - \partial_t(u_{t}+\alpha u)= \sin u+\gamma ,
 &\quad (x,t)\in \Omega, \vspace{2mm}  \\ 
    u(x,0)=h_0(x), \ \ \    u_t(x,0)=h_1(x),  &\quad   x\in [0,\ell],
  \vspace{2mm}  \\
      u_x(0,t)=\varphi_0(t), \ \ \  u_x(\ell,t)=\varphi_1(t), &\quad   0\leq t \leq T.
   \end{array}
     \right.
 \end{equation}

When $\varepsilon \equiv 0$, problem (\ref {21}) turns into the  Neumann   problem  related to operator ${\cal L}_0 $ defined in  (\ref{n19}) with  {\em the same initial boundary  conditions}:   

  \begin{equation}          \label{n11}
  \left \{
   \begin{array}{ll}
\displaystyle
      U_{xx} -
 \partial_t(U_{t}+\alpha\,U)\,=\sin U\ +\gamma &\quad(x,t)\in \Omega,\vspace{2mm}\\ 
    U(x,0)=h_0(x), \  \  \  U_t(x,0)=h_1(x), &\quad x\in [0,\ell],\vspace{2mm}  \\
     U_x(0,t)=\varphi_0(t), \  \ \  U_x(\ell,t)=\varphi_1(t), &\quad  \ 0 \leq t \leq T.
   \end{array}
  \right.
 \end{equation}

\noindent 
The influence of the dissipative term on the wave behaviour of $\,U\,$ can be  estimated when   the difference  $ d $ defined  in (\ref{n10})  is  evaluated.
 
  So, let us consider  the following  problem    related to the {\em  remainder} term $ \, d:\,$

\begin{equation}          \label{n24}
  \left \{
\begin{array}{ll}
     \partial_{xx} \, \,(\varepsilon
\partial_t \,+1)\,d - \partial_t(\partial_t+\alpha\,)d\,=\,F(x,t,d),\        & (x,t)\in \Omega,\vspace{2mm}  \\ 
    d(x,0)=0, \ \  \    d_t(x,0)=0, \  \  & x\in [0,\ell], \vspace{2mm}  \\ 
     d_x(0,t)=0, \  \ d_x(l,t)=0,  &  0\leq t \leq T
 \end{array}
  \right.
 \end{equation}

with  
 
 \begin{equation}                   \label{25}
\,\,{F}(x,t,d)= \sin (d+U) - \sin U - \varepsilon \,  U_{xxt}.\,
\end{equation}  
 
 If one  assumes:

 \begin{equation}  \label{26}
 \left \{                                             
\begin{array}{ll}
  &\gamma_n=\frac{n\pi}{\ell},\qquad h_n=\frac{1}{2}(\alpha+\varepsilon\gamma_n^2),\vspace{2mm}  \\ \\ & \omega_n=\sqrt{h_n^2-\gamma_n^2}
\\  \\
& H_{n}(t)=\,\, e^{-h_nt}\,\,  \frac{\sinh\,(\omega_nt)}{\omega_n}\,\
\end{array}
  \right. 
 \end{equation}

by means of standard techniques,   the Green function  $ G=G(x,t) $  of problem (\ref{n24}) is given by

 \begin{equation}\label{27}                                               
G(x,t,\xi)= \frac{1}{\ell}\,\, \frac{1- e^{-\alpha \, t }}{ \,\alpha}\,\, \,\,+ \,\frac{2}{\ell}\,\,\sum_{n=1}^{\infty}
H_{n}(t) \,  \,  \cos\gamma_n\xi\,  \, \cos\gamma_nx. 
\end{equation}

Moreover, by means of standard methods related to integral equations and thanks to  the fixed point theorem, owing to basic properties of  the Green function $ G  $  already proved in \cite{mda}, it is possible   to prove   that problem (\ref{n24})  admits a unique regular solution in $ \Omega $ and it results (see, f.i \cite{ddf,c,dmm}):

\begin{eqnarray}  \label{28}
 d(x,t)=\,\,  -\,\frac{1}{\ell} \int_{0}^{t} \, d\tau\,\, \int_0^\ell \,\,\bigl[\frac{1- e^{-\alpha \, (t-\tau) }}{\alpha} \bigr]\,\, F(\xi,\tau,d(\xi,\tau)) \, d\xi\ \\ \nonumber \\  \nonumber    - \,\,\frac{2}{\ell}\,\int_0^ t d\tau\, \int_0^\ell\, H(x,\xi,t-\tau)\,
F(\xi,\tau,d(\xi,\tau))\,\,d\xi.
\end{eqnarray}

where 
\begin{equation}  \label{29}
H= \sum_{n=1}^{\infty}
H_{n}(t) \,  \,  \cos\gamma_n\xi\,  \, \cos\gamma_nx.
\end{equation}

\section{ Estimates related to the Green Function}

In order to obtain an estimate of  function $ d(x,t,\varepsilon),  $ it appears necessary to evaluate the Green Function. 

So that, let us assume $ 0< \alpha <1  $ and  $ 0<\varepsilon <1 $ and  let us denote by $ N_1 $ the minimum  natural number larger than $\frac{\ell}{ \pi \varepsilon}\,\,(\, 1 - \sqrt{ 1-  a \, \varepsilon }\,), $ and by $ N_2  $ the maximum natural number smaller than  $\frac{\ell}{ \pi \varepsilon}\,\,(\, 1 + \sqrt{ 1-  a \, \varepsilon }\,).\,\, $

It is necessary to distinguish  two cases:  when   $ N_1 <1 $  or $ N_1>1. $  If $ N_1 <1, $  function  $ H_n(t) $  in $ (\ref{26})_3 $ contains trigonometric functions for  $N_1 \leq n\, \leq\, N_2 $  and hyperbolic terms for $ n\geq N_{2}+1. $  When $ N_1 >1, $ interval $ [1,  N_1-1]$ with  hyperbolic terms must be considered,too.

Firstly, in the hypothesis of $ N_1>1, $ we analyze hyperbolic terms both in $ [1, N_1-1] $  and   in $ [N_2+1, \infty] $ proving the following lemma:

\begin{lemma} \label{th1}

Whatever  $ 0< \alpha <1  $ and  $ 0<\varepsilon <1 $  may be,  there exist   positive constants $ B_i,\,(i=1,2)$   independent from $ \varepsilon,  $ such that:

\begin{equation} \label{31}
\sum_{n=N_{2}+1}^{\infty}  e^{-h_n t} \,\frac{\sinh (\omega_n \,t)}{\omega_n} \leq  B_1 e^{\,-\frac{\,t}{4\varepsilon} },
\end{equation}

\begin{equation}\label{210}
\sum_{n=1}^{N_1-1}  e^{-h_n t} \,\frac{\sinh (\omega_n \,t)}{\omega_n} \leq  \,B_2 \,e^{\frac{- t\,}{2 (\varepsilon \pi^2/\ell^2)}} \qquad  (N_1>1).
\end{equation}

\end{lemma}

\begin{proof}

Let us start with  considering the infinite interval $ (N_2+1, \infty). $  
If     $ c $  denotes  an  arbitrary positive constant less than $ 1, $  let $ N_c   $  be  the integer part of $ \ell/(\pi  \varepsilon \sqrt{c}) (1+ \sqrt{1-a\varepsilon \,c}).  $  

Interval $( N_2+1, \infty ) $ will be divided into intervals $ (N_{2}+1, N_{c}-1); ( N_{c}, \infty)   $ and it will result:
\begin{eqnarray}
 \label{211}
\nonumber &\sum_{n=N_{2}+1}^{\infty}  e^{-h_n t} \,\frac{\sinh (\omega_n \,t)}{\omega_n} =
\nonumber\\
\\
\nonumber &\sum_{n=N_{2}+1}^{N_c-1}  e^{-h_n t} \,\frac{\sinh (\omega_n \,t)}{\omega_n}+ \sum_{n=N_{c}}^{\infty}  e^{-h_n t} \,\frac{\sinh (\omega_n \,t)}{\omega_n} = H' + H''
\end{eqnarray}

In order to evaluate $ H' $ we state that

\begin{equation} \label{212}
e^{-h_n t} \,\,\, \frac{\sinh (\omega_n \,t)}{  \,\omega_n}= e^{-(h_n-\omega_n)t }\int _0^t e^{-2\omega_n\tau}  d\tau \,\,\leq  t \,\,e^{-(h_n-\omega_n)t }
\end{equation}

and
\begin{equation} \label{213}
h_n-\omega_n = h_n - h_n \sqrt{1-\frac{\gamma_n^2}{h_n^2}} \geq  h_n - h_n \bigg(1-\frac{\gamma_n^2}{2 h_n^2}\bigg).
\end{equation}

So that, since $ n\geq N_2+1 > \ell/( \pi\varepsilon) $ and $ a\varepsilon <1, $ it results:

\begin{equation} \label{214}
\,e^{-(h_n-\omega_n)t }\leq  \,\, e^{- t\,\,\frac{\gamma_n^2}{2 h_n}}\leq  e^{- \,\frac{t}{2\varepsilon }}\,.
\end{equation}

Moreover, taking into account that

 \begin{equation} \label{a}
 e^{-x} \leq \frac{\nu^\nu }{(ex)^\nu } \qquad  \qquad \forall \,\,x\, > 0; \, \,\,\forall \nu >0\,,
 \end{equation}
 
 one has: 

\begin{equation} \label{216}
H' \leq\sum_{n=N_{2}+1}^{N_c-1} t e^{- \,\frac{t}{2\varepsilon}}\,  \leq B'  \,\,e^{- \,\frac{t}{4\varepsilon}}\,\,
\end{equation}

\noindent where   $ B' =\frac{4\varepsilon}{e}  (N_c-N_2-1)   \leq \frac{8 \ell}{e\,\pi \sqrt{c}}.$ 

As for $ H''\,,\,\, $  when  $  n \geq N_c,$ one has  $\gamma_n \geq \frac{1}{ \varepsilon \sqrt{c}} ( 1+ \sqrt{1- a \varepsilon c})\Rightarrow\varepsilon \sqrt{c} \,\,{\gamma_n}^{2}  + a  \sqrt{c} - 2 {\gamma_n} \geq 0, $ so that $ \frac{2{\gamma_n}}{ a+\varepsilon \gamma_n^2} \leq\sqrt{c}\Leftrightarrow\frac{{\gamma_n}}{h_{n}} \leq\sqrt{c}.   $ As a  consequence it results:

\begin{equation}
\omega_n= h_n \sqrt{1-\frac{\gamma_n^2}{h_n^2}}   \geq  h_n\sqrt{1-c}\, \geq \frac{\varepsilon \pi^2  n^2 \sqrt{1-c}\,\,}{2\ell^2  }.
\end{equation} 

So,   if $ 0<\eta<1, $  denoting by  $ \zeta(x)   $   the Riemann zeta function and since $ \,N_c\, \geq \ell/( \varepsilon\,\pi\sqrt{c})   $ one has

\[  H'' \leq  \sum_{n=N_c}^{\infty} e^{\,-\frac{t}{2\varepsilon}}\frac{4 \ell^2}{\pi^2\sqrt{1-c}}\,\,\frac{1}{\varepsilon n^{1-\eta}}\frac{1}{ n^{1+\eta} } \,\leq  \frac{4\ell^2 {(\pi \sqrt{c}/\ell)}^{(1-\eta)}}{ \pi^2\sqrt{1-c}} \,\,\varepsilon ^{-\eta} \,\,\zeta(1+\eta)\, e^{\,-\frac{t}{2\varepsilon}} \]

\noindent and considering  (\ref{214})  as well,  if $ B'' = \frac{ 8\, \ell^2\,{(\,\pi \sqrt{c}/\ell)}^{(1-\eta)}}{ e\,\pi^2 \sqrt{1-c}} \, \,\,\zeta(1+\eta),  $ since  $ t\geq 1, $ it results:

\begin{equation} \label{n26}
 H''  \leq B'' \,\, \varepsilon^{1-\eta} e^{\,-\frac{t}{4\varepsilon}}.   
\end{equation}

Finally, according to  (\ref{212}), (\ref{213}) and (\ref{a}),  one has:

\begin{equation} \label{234}
\sum_{n=1}^{N_1-1}  e^{-h_n t} \,\frac{\sinh (\omega_n \,t)}{\omega_n} \leq\sum_{n=1}^{N_1-1} t   e^{- t\,\frac{\gamma_n^2}{2 h_n}}  \leq   B_2 e^{-\frac{ t}{2(\varepsilon \pi^2/\ell^2)}}  
\end{equation}

 \noindent where,    $ B_2 = ( N_1-1)  \frac{2\varepsilon\pi^2}{ e\ell ^2} \leq   \frac{4\pi a}{ e\ell \pi}. $
 
 \noindent Therefore, since (\ref{216}), (\ref{n26}), and (\ref{234}) hold, denoting by  $ B_1= \max \{B', B''\}, $ the  theorem is proved. 
\end{proof}

Now,    letting 

 \begin{eqnarray}
 &\nonumber \omega_0=\sqrt{\gamma_n^2-\frac{\alpha^2}{4}};  \qquad \tilde \omega_n=\sqrt{\gamma_n^2-h_n^2}; \qquad h_1= \frac{1}{2}\big(\alpha + \varepsilon\frac{\pi^2}{\ell^2}\big)
\nonumber \\
 \\
\nonumber  &  R(n,t)=e^{- h_n \,t} \frac{\sin (\tilde \omega_n \,t)}{\tilde \omega_n}-  e^{-h_1 t} \,\,\frac{\sin ( \omega_0 \,t)}{ \omega_0}  
 \,
\end{eqnarray}
  
\noindent in order to evaluate circular terms, we will firstly prove the following inequality involving difference $R(n,t)$. Then, attention will be devoted to all the terms of series in $ H_n(t). $ 

 \begin{lemma} \label{th2}

Whatever  $ 0< \alpha <1  $ and  $ 0<\varepsilon <1  $  may be,  there exists a  positive constant $ B_3$ independent from $ \varepsilon,  $ such that,  letting   $ 0<\eta<1\,,$   $ 0<h<1/2,$  $ 1<k<3/2,  $  it results:

\begin{equation}  \label{313}
\sum_{n=1}^{N_2} [\,e^{- h_n \,t} \frac{\sin (\tilde \omega_n \,t)}{\tilde \omega_n}-  e^{-h_1 t} \,\,\frac{\sin ( \omega_0 \,t)}{ \omega_0}
 \,]  \leq  B_3 \,\, \varepsilon^\rho \,\,e^{-\frac{\alpha\,t}{4} }
 \end{equation}
 
\noindent where $ \rho = min\{ \eta, \, 1-\eta, \,1/2-h ,\,3/2-k\}. $   
 \end{lemma}
 
 \begin{proof}
 \noindent   Indicating by $ \hat R (n,s) $   the Laplace transform of function $ R(n,t), $ one deduces that  

\[\hat R(n,s)= \frac{(s+h_1)^2+\omega_0^2 -(s+h_n)^2-\tilde \omega_n^2}{[(s+h_1)^2+\omega_0^2] \,\,\,[(s+h_n)^2+\tilde\omega_n^2}
 \]

 So that, denoting by $ \tilde g= \frac{\pi^2}{4\ell^2}(2\alpha+\frac{\pi \varepsilon}{\ell})  ,
$   it results

\begin{eqnarray} \label{222}
 &\nonumber R(n,t) = \frac{ \varepsilon\tilde g   }{ \,\omega_0 \tilde \omega_n }   \int _0^t  e^{-h_1(t-\tau)} \sin (\omega_0 (t-\tau))  e^{-h_n \tau}   \sin (\tilde \omega_n \,\tau) d\tau  +
 \\
 \\
 &\nonumber\frac{ \pi^2/\ell^2\,\,\varepsilon   (n^2-1)}{\omega_o} \int_0^t    e^{-h_1(t-\tau)} \sin (\omega_0 (t-\tau))  \{ e^{-h_n \tau} [\frac{h_n}{\tilde\omega_n}  \sin (\tilde\omega_n  \tau)-  \cos (\tilde\omega_n \,\tau)]\}\,.
\end{eqnarray}

Indicating by

\begin{equation}
  g_0= \sqrt{\frac{\pi^2\,}{\ell^2}-\frac{a^2}{4}} \quad \,g= \frac{\pi^2}{4\ell^2}(2\alpha+\frac{\pi }{\ell})
\end{equation}

one has:
 
 \begin{equation}  \label{224}
  \omega_0 \geq n\, g_0,   \qquad \tilde g \leq g.   
\end{equation}

 \noindent Cases $ n=1, $ $ n=N_2, $  and $ 2 \leq n\leq N_2-1 $ will be considered and  the value of function $R(n,t) $ in correspondence of $ n\equiv\bar n $  will be denoted by $ R(\bar n, t).  $

 When  $ n=1,  $ from (\ref{222}),       if  $ g_1 = 32  g  /[(\alpha\, e\,)^2 g_0],$ one has: 
\begin{equation}\label{318}
 R(1,t) \leq   e^{-\frac{\alpha\, t}{4}} \, \varepsilon  \,\, g_1 
\end{equation}

Moreover, since

\[ h_n-h_1 = \frac{\pi^2(n^2-1) \varepsilon }{2\ell^2},\] 

for $ n= N_2,  $  denoting by $ \eta  $ a constant such that   $0<\eta <1, $ let

\begin{eqnarray}
&g_{N_2}=   g\,/(2\,g_0 \, \varepsilon N_2\,)  
\\
\nonumber\\
&A'  =   \frac{4}{\varepsilon N_2 g_0}\bigg[\bigg( \frac{\alpha}{2} + \frac{\pi^2 \varepsilon^2 N_2^2}{2\ell^2} \bigg) \,\bigg(\frac{1+\eta}{\sqrt{1-a}}\bigg)^{1+\eta}+ \bigg  (\frac{\eta}{\sqrt{1-a}}\bigg)^{\eta}\bigg],
\end{eqnarray} 

 and $   $ considering also that

\[ 2\sqrt{1-a} \,\leq \,\varepsilon^2 \frac{\pi^2}{\ell^2} (N_2^2-1)\,\leq 4,  \,   \] 

\noindent  
by means of  (\ref{a}), one obtains: 

\begin{equation}
 R(N_2,t) \leq e^{-\frac{a t}{2}}  \varepsilon^2 t^2\,\, g_{N_2}+ A'\, e^{-\frac{a t}{2}} t^{1-\eta}\,\varepsilon ^\eta.
\end{equation}

So that, there exists a positive constant $ A'' $ independent from $\varepsilon $ such that:

\begin{equation}  \label{320}
 R(N_2,t)  \leq  A''\, \,\varepsilon ^\eta\,\,e^{-\frac{a t}{4}} 
\end{equation}

Now, let us consider $ 2 \leq n\leq N_2-1. $ Since  $ \tilde \omega_n\,\geq  \gamma_n \sqrt{1-h_n/\gamma_n}= n  \,
\Phi_n, \, $  indicating by

\begin{eqnarray}
\nonumber  & s = \sqrt{\frac{\pi}{2\ell}} \sqrt{\frac{2\sqrt{1-a\varepsilon}-\varepsilon \pi /\ell}{1-\varepsilon \pi/\ell +\sqrt{1-a\varepsilon}}} = \frac{( \Phi_n )_{n= N_2-1}}{\sqrt{\varepsilon }}, \,\,
\nonumber\\
 \\
\nonumber &p=\frac{\sqrt{4\ell(\pi-\varepsilon \ell)-\pi^2 a }}{2\sqrt{\pi \ell}} =( \Phi_n )_{n=2}
\\ \nonumber \mbox{and by } &q= \min \,\{ \,s, p\,\}, 
\end{eqnarray}

 \noindent   it is possible to choose  a positive constant  $\, g_2 \, $ depending on parameter  $ a, $ but  independent from $ \varepsilon, $  such that $ g_2 \leq  q\,.\,  $ In such a way, for all $ 2\leq n\leq  N_2-1 $ one has:

\begin{equation}  \label{227}
 \tilde \omega_n\, \geq n \,\sqrt{\varepsilon} \,\, g_{2}.      
\end{equation}

So that, it results:

\begin{equation} \label{322}
 e^{\frac{a\,t}{2} }\sum_{n=2}^{N_2-1} R(n,t) \leq   \sum_{n=2}^{N_2-1}\frac{g}{ n^2 g_0 g_2}\sqrt{\varepsilon }t + 
\end{equation}

\[  \sum_{n=2}^{N_2-1} \bigg( \frac{\pi^2\sqrt\varepsilon}{\ell^2 g_0g_2} h_n+ \frac{\pi^2\varepsilon n}{\ell^2 g_0}  \bigg) \int_0^t   | \sin (\omega_0 (t-\tau))|   e^{- (h_n-h_1 )\tau} d\tau.\]

  Considering      $ 0<\eta <1,$ and indicating by $ h, $ and  $ k $ two constants such that  
 
 \[ 0<h<1/2 \qquad  1< k< 3/2,\]
 
 let 
 
 \[ C_1= \frac{ \pi^2\alpha}{ \ell^2 g^2_0g_2}\bigg(\frac{2 h \ell^2}{e\,\pi^2} \bigg)^h , \quad C_2 =\frac{\pi^4}{\ell^4g^2_0 g_2}\, \bigg(\frac{2 k\ell^2}{e\,\pi^2} \bigg)^k  \qquad C_3= \frac{2\pi^2}{\ell^2 g_0^2}\, \bigg(\frac{2 \eta\ell^2}{e\,\pi^2} \bigg)^\eta  . \]

The analysis lead us to consider asymptotic behavior when variable $ t  $ tends to infinity. Therefore it does not effect generality if we consider $ t\geq 1. $ In this hypothesis and since  $    \int_0^t   | \sin(\omega_0 (t-\tau))| d\tau \,\leq 2/\omega_0,\,\,$ one obtains:

\begin{equation}
 e^{\frac{a\,t}{2} }\sum_{n=2}^{N_2-1} R(n,t) \leq   \sum_{n=2}^{N_2-1}  \frac{g}{ n^2 g_0 g_2}\sqrt{\varepsilon }t 
 \end{equation}
 
 \[+ \sum_{n=2}^{N_2-1}\bigg[  C_1 \,\frac{\varepsilon^{1/2-h}}{n(n^2-1)^h}+  C_2 \,\frac{n\,\,\varepsilon^{3/2-k}}{\,(n^2-1)^k } + C_3 \frac{\,\varepsilon^{1-\eta}}{\,(n^2-1)^\eta } \bigg].\]

So, it is possible to choose, in this case as well, a positive constant $C_4  $   independent from $ \varepsilon,  $  such that:

\begin{equation}\label{231}
\sum_{n=2}^{n=N_2-1} R(n,t)  \leq  C_4 [\sqrt{\varepsilon} t  +\,\,\varepsilon^{1/2-h}+\varepsilon^{3/2-k} +\varepsilon^{1-\eta}]e^{-\frac{a\,t}{2} } .
\end{equation}

For this  and since (\ref{318}) and (\ref{320}), (\ref{313}) holds.

 \end{proof}
 
\textbf{Remark}
 Naturally, by tracing back this demonstration, it  is possible to prove that for $ N_1>1 $ the  Lemma \ref{th2}  holds   even if  interval $ [ 1, N_2] $ has to be replaced with  $ [ N_1, N_2]. $

So that,  for any value of $ N_1, $  the following theorem can be proved:

\begin{theorem} \label{th3}

Whatever  $ 0< a <1  $ and  $ 0<\varepsilon <1 $  may be, there exists a  positive constant $ B , $   independent from $ \varepsilon,  $ such that:

\begin{equation}\label{n42}
\sum_{n=1}^{\infty}  e^{-h_n t} \,\frac{\sinh(\omega_n \,t)}{ \omega_n} \leq\,B \, e ^{-{m\,\, t }}
\end{equation}

\noindent being $m= max\{a/4,a \ell^2/(2 \pi^2)\}. $

\end{theorem}

\begin{proof}

Hyperbolic terms have been evaluated   by means of lemma  (\ref{th1}).  Besides, if  $ N_1 <1, $ it  remains to be proven that there exists a positive  constant  $ B_4 $ independent from $ \varepsilon $ such that:
 
 \begin{equation}\label{325}
\sum_{n=1}^{N_2}  H_n = \sum_{n=1}^{N_2}  e^{-h_n t} \,\frac{\sin (\tilde\omega_n \,t)}{\tilde \omega_n} \leq\,B_4 \, e ^{-{\,\frac{a\,t}{4} }}
\end{equation}
 But  (\ref{325}) follows on application of  lemma \ref{th2} and since  

\begin{equation}  \label{326}
\sum_{n=1}^{N_2} \,H_n(t) \leq  R(n,t)  + e^{-\frac{at}{2}}\,\ \sum_{n=1}^{\infty} 
 H_n^o(t)\,   
 \end{equation}
 
\noindent where $ H^o_n(t)= \frac{\sin(\omega_0t)}{\omega_0}. $

Similar estimates hold also for $ N_1>1, $ namely when  circular terms appear only for  $ n\in [N_1,N_2]. $ 

So that,  since   $ e^{\,-\frac{\,t}{4\varepsilon} }\leq e^{\,-\frac{a\,t}{4} } $ and $ e^{-\frac{ t}{2(\varepsilon \pi^2/\ell^2)}}  \leq e^{-\frac{t a}{2 \pi^2/\ell^2}},  $ from (\ref{31}), (\ref{210})  and   (\ref{325}), theorem is proved.

\end{proof}

\section{Explicit solutions of the hyperbolic equation }

Let us consider the semilinear  second order equation:  

\begin{equation}  \label{51}
  U_{xx} \,- \,U_{tt} \, - \,  \alpha \,U_{t} =   \sin U \,  \ + \gamma 
\end{equation}

When $ \alpha =   \gamma  = \, 0  $,  (\ref{51}) represents  the  sine - Gordon equation  and  there exists  a wide  literature   about its   classes of soliton solutions \cite{jsb,wj,adc}. Besides, if  $ \alpha \neq 0 $ or  $  \gamma \neq \, 0  $, other solutions are shown  in \cite{fgmm,zro}.

Now, let $\, f \,$ be  an arbitrary function, and let  us  consider  the following function $ \,\,\Pi (f)\,\,$:

\begin{equation}     \label{52}
\Pi(f) = \, 2 \, \arctan\,\, e^f \,\,  
\end{equation}

\noindent so that

\begin{equation}    \label{53}
\sin \,  \Pi (f) \,\, = \frac{1}{\cosh (f) }, \qquad  \cos \,  \Pi (f) \,\, = -   \tanh (f) .
\end{equation}

By means of function (\ref{52}) it is possible to find a class of solutions of  equation  (\ref{51}).

 Indeed, it is possible  to  verify   that the following function:

\begin{equation}
U \, = 2 \, \Pi [f(\xi)]\,  \qquad \mbox{with}\qquad \xi =\, \frac{x-t}{\alpha}
\end{equation}

 is a solution of (\ref{51}) provided that one has:

\begin{equation}   \label{55}
-\alpha\, U_{t}\, = \, \sin U \, +\, \gamma.
\end{equation}

\vspace{3mm} Moreover,  since (\ref{53}) and   being $\,\, \dot \Pi = \frac{1}{\cosh f}, \,\,\,  $ 
 it results:

\[ \displaystyle -\alpha\, U_{t}\, = 2\frac{f'}{\cosh f}; \,\,\qquad \sin U \,\,=  \,-\,\,2 \,\frac{\tanh f }{\cosh f}. 
\]

\vspace{3mm} So, from (\ref{55}), one deduces that function $ f $ must satisfy the following equation:

\begin{equation}     \label{56}
\frac{df}{-\tanh f \, + \gamma/2\, \cosh f}\, = d\, \xi.
\end{equation}

According to value of constant  $ \gamma, $ it is possible to   find explicit classes of solution (see f.i. \cite{df213}), and we will consider as an example    $ \gamma=0  $ and  $ \gamma=1.  $   

So,  indicating   by $ f_{0},\xi_{0} $ the respectively  initial data of function $ f $ and variable $ \xi, $  when $ \gamma =0,  $  for $ r_0= e^{\xi_{0}} \, \sinh f_{0},  $  since (\ref{56}) it results:

 \begin{equation}
 \sinh f = \,   \,r_0\ e^{-\xi}, 
 \end{equation} 
 
  and one has:

 \begin{equation}    \label{AAA}
U\,=\, 4 \arctan (\, y\,+\,\sqrt{y^2 \, +\, 1}\,) \quad \quad\mbox{with}\,\,\    y= \,r_0\, e^{-\, \xi\,}\, .\,
\end{equation}

Besides, if $ \gamma =1, $ since  $ \int \frac{df}{1/2 \cosh  f - \tanh f}\,=  \, \frac{-2}{\sinh f -1} + const, $ denoting by  $ r_0= \xi_{0}+ \frac{2}{\sinh f_{0} -1}, $ it results:

\begin{equation}    \label{aaAAA}
U\,=\, 4 \arctan (\, y\,+\,\sqrt{y^2 \, +\, 1}) \quad \mbox{with} \quad y= \frac{2+ r_0-\xi }{r_0-\xi} \,.
\end{equation}

Our aim is to find a solution of problem (\ref{n11}) whose derivatives are bounded. In order to do so, let us consider the following initial boundary problem:

  \begin{equation}   
   \left \{
   \begin{array}{ll}
\displaystyle
      U_{xx} -
 \partial_t(U_{t}+\alpha\,U)\,=\sin U\  &\quad(x,t)\in \Omega,\vspace{2mm}\\ \displaystyle
     U(x,0)=2\arctan(e^{x/\alpha}), \  \  \  U_t(x,0)=-\frac{2}{\alpha} \,\,\frac{e^{x/\alpha}}{1+e^{2x/\alpha}}, &\quad x\in [0,\ell],\vspace{2mm}  \\ 
         U_x(0,t)=\frac{2}{\alpha}\,\, \frac{e^{-t/\alpha}}{1+e^{-2t/\alpha}}, \  U_x(\ell,t)=\frac{2}{\alpha}\,\, \frac{e^{(\ell-t)/\alpha}}{1+e^{2(\ell-t)/\alpha}}, &\quad  \ 0\leq t \leq T,
   \end{array}
  \right.
 \end{equation}    
 
 which  admits the solution

\begin{equation}
U = 2 \arctan e^{(x-t)/\alpha}.
\end{equation}

It results:

\begin{equation} \label{488}
 U_{xxt}= -\frac{2}{\alpha^3} \,\frac{e^{2\xi}-6+ e^{-2\xi}}{(e^{\xi}+e^{-\xi})^3}
 \end{equation} 

which is bounded  for all $ (x,t) \in \Omega. $

\section{Estimates for the remainder term}

Let $ U  $  be  a solution of the  reduced  problem (\ref{n11}) and let us assume

\begin{equation}
S(t)\, = \sup_{0\leq  x \leq \ell }\, | d(x,t)|. 
\end{equation}

It is possible to state:

\begin{theorem}

 \label{remainder}

If there exists   a  positive constant   $ b  $ such that 

\begin{equation} \label{61}
 |U_{xxt}\,(x,t)|\, \leq \,b, 
\end{equation}

\vspace{3mm}
\noindent then, indicating by $ T_{\varepsilon} = \log(1/\varepsilon^k), $ with $ 0<k<\alpha, $ for all $1\leq  t <  T_ \varepsilon  $   there exists a   positive constant  $\Gamma, $ independent from $ \varepsilon, $    such that:

\begin{equation}   \label{64}
0\,\leq \,S(t)\,\leq\, \Gamma \,\,\varepsilon ^{1-k/\alpha}\log(\frac{1}{\varepsilon^k}).
\end{equation}


\end{theorem}

\begin{proof}: Let us consider function  $ F $ defined in (\ref{25}):

\begin{equation}  \label{65}
\, F \,= \, \sin  \, (d+U) \, -\sin U  \,- \, \varepsilon \,U_{xxt}.   
\end{equation}

\vspace{3mm} According to (\ref{61}), function $ F(x,t,u) $  satisfies the following inequality: 

\begin{eqnarray}
&|F(x,t)| \leq \,\, |d(x,t)| +\varepsilon \,\,b.
\end{eqnarray}
 \noindent  So that, by means of (\ref{28})-(\ref{29}) and (\ref{n42}) and since   $ t\leq \log(\frac{1}{\varepsilon ^k}), $ one obtains:

 \begin{eqnarray}  \label{41}
d(x,t)\leq \,\,   C \,\varepsilon \,\,\log(\frac{1}{\varepsilon ^k}) +   \int_0^ t  \big[\frac{1}{\alpha}  +B e^{- m (t-\tau)}\,\big]
|d(\tau,\xi))|\,d\tau\, .
\end{eqnarray}

\noindent being $ C= (B+\frac{1}{\alpha})  \,b. $ 

So, considering that: 

\[   \int_0^ t  \big[\frac{1}{\alpha}  +B e^{- m (t-\tau)}\, \big]    d\tau\,\,\leq  \frac{1}{\alpha}  \log(\frac{1}{\varepsilon ^k}) +\frac{B}{m}, \]
by means of    a linear generalization of Gronwall's inequality( see f.i.(\cite{mpf}), it results:

\begin{eqnarray}  \label{41}
 d(x,t)\leq \,\,  C  \,\,\varepsilon \,\,\log(\frac{1}{\varepsilon ^k})  \,\, e^{ \frac{B}{m}} \frac{1}{\varepsilon^{k/\alpha}} = \Gamma \,\,\varepsilon^{1-k/\alpha} \,\,\log(\frac{1}{\varepsilon ^k})  \,\,  
\end{eqnarray}

from which theorem is proved. 

\end{proof}

\textbf{Remark} Estimate   (\ref{64}) specifies the infinite time-intervals where the effects of diffusion vanish when $ \varepsilon\rightarrow 0.  $  Indeed, the evolution of the superconductive model  is characterized by  diffusion effects which are of the order of $ \,\displaystyle  \varepsilon^{1-k/\alpha} \,\,\log(\frac{1}{\varepsilon ^k})\,$   with $ 0<k<\alpha. $

\textbf{Remark}
  Formula (\ref{488}) shows that  the class of functions    satisfying  hypotheses  of Theorem \ref{remainder}  
 is not empty.


%

\end{document}